\documentclass[conference]{IEEEtran}

\usepackage{amsthm}
\usepackage{times}
\usepackage{breqn}
\usepackage{algorithm}
\usepackage{algpseudocode}
\algrenewcommand\algorithmicrequire{\textbf{Input:}}
\algrenewcommand\algorithmicensure{\textbf{Output:}}

\usepackage{epsfig}
\usepackage{epstopdf}
\usepackage{subfigure}
\usepackage{multirow}
\usepackage{comment}
\usepackage{listings}
\usepackage{xcolor}
\usepackage{soul}
\usepackage{xspace}
\usepackage{url}

\newcommand{\bigO}[1]{$\mathcal{O}\paren{#1}$\xspace}
\newcommand{\bigOmega}[1]{\mbox{\mbox{$\Omega$}$\paren{#1}$}\xspace}
\newcommand{\bigTheta}[1]{\mbox{\mbox{$\Theta$}$\paren{#1}$}\xspace}
\newcommand{\paren}[1]{\left(  #1 \right)}
\newcommand{\edge}[1]{(  #1 )}

\newtheorem{theorem}{Theorem}

\newtheorem{lemma}{Lemma}

\newtheorem{definition}[theorem]{Definition}

\usepackage{cite}
\usepackage{amsmath,amssymb,amsfonts}
\usepackage{graphicx}
\usepackage{textcomp}
\usepackage{xcolor}
\def\BibTeX{{\rm B\kern-.05em{\sc i\kern-.025em b}\kern-.08em
    T\kern-.1667em\lower.7ex\hbox{E}\kern-.125emX}}

\graphicspath{{figures/}}

\makeatletter
\newcommand{\linebreakand}{%
  \end{@IEEEauthorhalign}
  \hfill\mbox{}\par
  \mbox{}\hfill\begin{@IEEEauthorhalign}
}
\makeatother

\begin{document}

\title{Triangle Counting Through Cover-Edges}
\author{

\IEEEauthorblockN{David A. Bader$^*$, Fuhuan Li, Anya Ganeshan$^+$, Ahmet Gundogdu$^\#$, Jason Lew, Oliver Alvarado Rodriguez, Zhihui Du}
\IEEEauthorblockA{\textit{Department of Data Science} \\
\textit{New Jersey Institute of Technology}\\
Newark, New Jersey, USA\\
\{bader,fl28,jl247,oaa9,zhihui.du\}@njit.edu}

}

\maketitle

\begingroup\renewcommand\thefootnote{$^+$}
\footnotetext{Anya Ganeshan attends Bergen County Academies, Hackensack, NJ}
\endgroup
\begingroup\renewcommand\thefootnote{$^\#$}
\footnotetext{Ahmet Gundogdu attends Paramus High School, Paramus, NJ}
\endgroup
\begingroup\renewcommand\thefootnote{$^*$}
\footnotetext{This research was partially supported by NSF grant number CCF-2109988.}
\endgroup
\begin{abstract}
\looseness=-1
Counting and finding triangles in graphs is often used in real-world analytics to characterize cohesiveness and identify communities in graphs. In this paper, we propose the novel concept of a cover-edge set that can be used to find triangles more efficiently. We use a breadth-first search (BFS) to quickly generate a compact cover-edge set. Novel sequential and parallel triangle counting algorithms are presented that employ cover-edge sets. The sequential algorithm avoids unnecessary triangle-checking operations, and the parallel algorithm is communication-efficient. The parallel algorithm can asymptotically reduce communication on massive graphs such as from real social networks and synthetic graphs from the Graph500 Benchmark. In our estimate from massive-scale Graph500 graphs, our new parallel algorithm can reduce the communication on a scale~36 graph by 1156x and  on a scale~42 graph by 2368x. 
\end{abstract}


\section{Introduction}
\label{s:introduction}
\emph{Triangle counting} \cite{al2018triangle} is a fundamental problem in graph analytics, which involves finding the number of unique triangles in a graph. It plays a crucial role in various graph analysis techniques such as clustering coefficients \cite{watts1998collective}, k-truss \cite{cohen2008trusses}, and triangle centrality \cite{burkhardt2021triangle}. The significance of triangle counting is evident in its application in high-performance computing benchmarks like Graph500 \cite{graph500} and the MIT/Amazon/IEEE Graph Challenge \cite{GraphChallenge}, as well as in the design of future architecture systems (e.g., IARPA AGILE \cite{slides_on_AGILE}).

Both sequential and parallel triangle counting algorithms have been studied extensively since 1977 \cite{itai1977}.  Latapy \cite{latapy2007practical} provides a comprehensive overview of sequential triangle counting and various finding algorithms. Existing techniques, including list intersection, matrix multiplication, and subgraph matching \cite{pandey2021trust}, are  techniques used to count triangles.

\looseness=-1
To enhance the performance of triangle counting, Cohen \cite{cohen2009graph} introduced a novel map-reduce parallelization technique that generates \emph{open wedges} between triples of vertices in the graph. It determines whether a closing edge exists to complete a triangle, thus avoiding the redundant counting of the same triangle while maintaining load balancing. Many parallel approaches for triangle counting \cite{pearce2017triangle,ghosh2020tric} partition the sparse graph data structure across multiple compute nodes and adopt the strategy of generating open wedges, which are sent to other compute nodes to determine the presence of a closing edge. Consequently, the communication time for these open wedges often dominates the running time of parallel triangle counting.

In traditional edge-based triangle counting methods, all triangles are identified by accumulating the sizes of intersections between pairs of endpoints for each edge. \emph{Direction-oriented} approaches can avoid counting the same triangle multiple times. However, in this paper, we propose a novel approach that efficiently identifies all triangles using a reduced set of edges known as a cover-edge set. By leveraging the cover-edge-based triangle counting method, unnecessary edge checks can be skipped while ensuring that no triangles are missed. This significantly reduces the number of computational operations compared to existing methods. Furthermore, for distributed parallel algorithms, the cover-edge-based method can greatly reduce overall communication requirements. As a result, our proposed method offers improved efficiency and scalability for triangle counting.

Our contributions include:
\begin{itemize}
\item A novel concept, \emph{Cover-Edge Set}, is proposed to support efficient triangle counting. The essential idea is that we can identify all triangles from a significantly reduced cover-edge set instead of the complete edge set. A simple breadth-first search (BFS) is used to orient the graph's vertices into levels and to generate the cover-edge set.
\item A novel triangle counting and finding algorithm, \emph{CETC}, is developed based on the concept of \emph{Cover-Edge Set}.  \emph{CETC} runs in $\mathcal{O}(m \cdot d_{\text{max}})$ time and $\mathcal{O}(n+m)$ space, where $d_{\text{max}}$ is the maximal degree of a vertex $v \in V$.
\item A novel communication-efficient distributed parallel algorithm for triangle counting and finding, \emph{Comm-CETC}, is also developed based on the concept of \emph{Cover-Edge Set}. \emph{Comm-CETC} can asymptotically reduce the communication to improve total performance.
\end{itemize}

The remainder of the paper is organized as follows. Section~\ref{sec:new_approach} presents our new approach for triangle counting. In Section~\ref{sec:parallel}, we employ our idea in distributed parallel triangle counting to reduce communication. Section~\ref{sec:related_work} discusses related work. Lastly, in Section~\ref{sec:conclusion}, we conclude the paper.

\section{Cover-Edge Set Based Triangle Counting}
\label{sec:new_approach}

\subsection{Notations and Basic Idea}
\label{subsec:notation}
Let $G = (V, E)$ be an undirected graph with $n=|V|$ vertices and $m = |E|$ edges. A \emph{triangle} in the graph is a set of three vertices $\{v_a, v_b, v_c\} \subseteq V$ such that $\{\edge{v_a, v_b}, \edge{v_a, v_c}, \edge{v_b, v_c}\} \subseteq E$. We will use $N(v) = \{ u|u \in V \wedge (\edge{v,u} \in E)\}$ to denote the \emph{neighbor set} of vertex $v \in V$.
The degree of vertex $v \in V$ is $d(v) = |N(v)|$, and $d_{\text{max}}$ is the maximal degree of a vertex in graph $G$.

\begin{definition}[Cover-Edge and Cover-Edge Set]\label{def:coveredge}
For any edge $e$ of a triangle $\Delta$ in graph $G$, $e$ is referred to as a cover-edge of $\Delta$. For a given graph $G$, an edge set $S \subseteq E$ is called a cover-edge set if it contains at least one cover-edge for every triangle in $G$.
\end{definition}

Based on the given definition, it is evident that the entire edge set $E$ can serve as a cover-edge set $S$ for graph $G$. However, our proposed method aims to efficiently count all triangles using a smaller subset of edges instead of $E$. Thus, the primary challenge lies in generating a compact cover-edge set, which forms the initial problem to be addressed in our approach. Let $k=|S|/|E|$. Our goal is to identify cover-edge sets with the smallest $k$. In this paper, we propose using breadth-first search (BFS) to generate a compact cover-edge set.

\begin{definition}[BFS-Edge]\label{def:bfsedges}
Let $r$ be the root vertex of an undirected graph $G$. The level $L(v)$ of a vertex $v$ is defined as the shortest distance from $r$ to $v$ obtained through a breadth-first search (BFS). From the BFS, we classify the edges into three types:
\begin{itemize}
\item \emph{Tree-Edges}: These edges belong to the BFS tree.
\item \emph{Strut-Edges}: These are non-tree edges with endpoints on two adjacent levels in the BFS traversal.
\item \emph{Horizontal-Edges}: These are non-tree edges with endpoints on the same level in the BFS traversal.
\end{itemize}
\end{definition}

\begin{figure}
    \centering
    \includegraphics[width=0.35\textwidth]{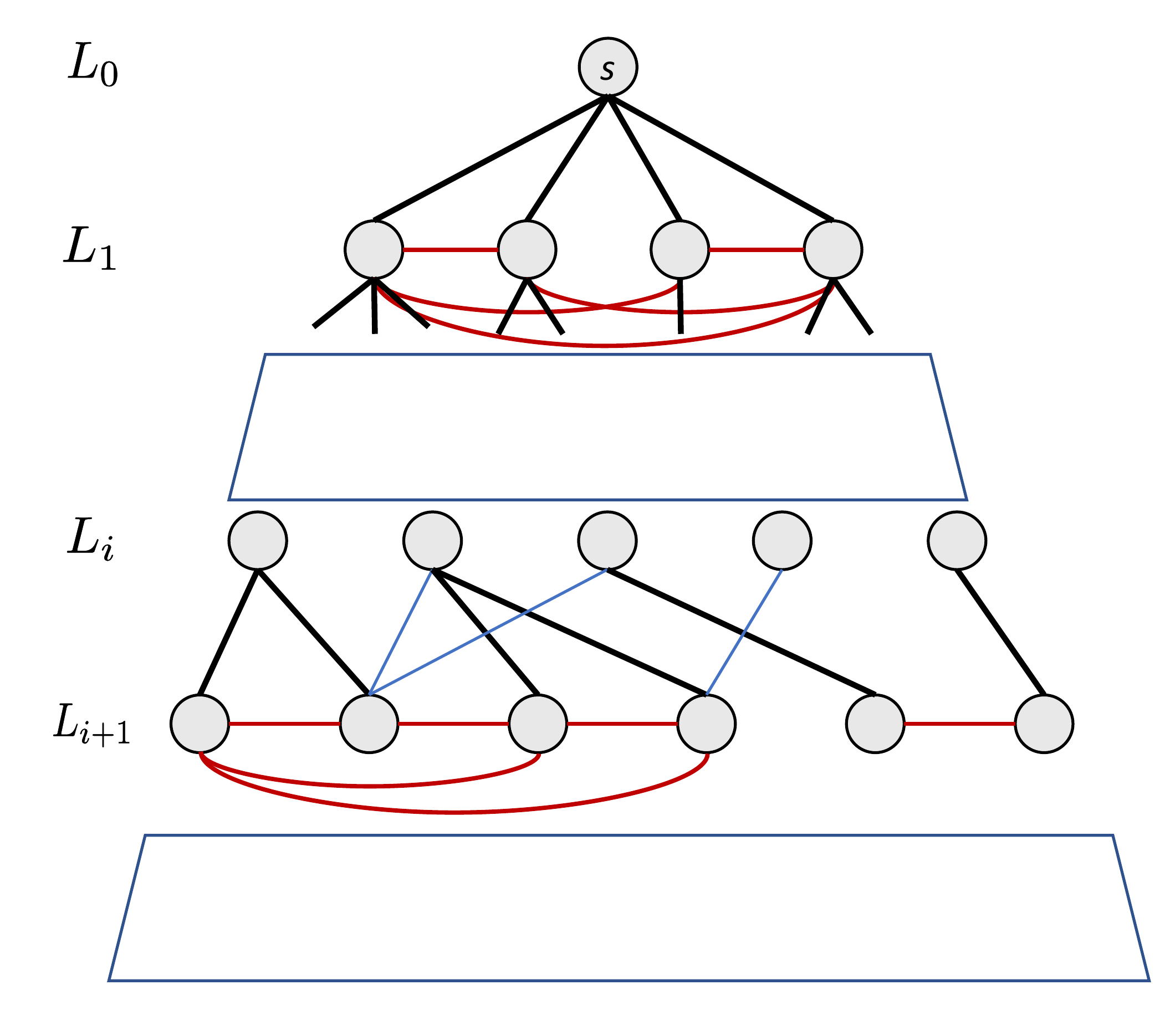}
    \caption{Example BFS tree of Graph $G$. The tree-edges are black, strut-edges are blue, and horizontal-edges are red.}
    \label{fig:bfs}
\end{figure}
Fig.~\ref{fig:bfs} gives an example of these different edge types. 

\begin{lemma}
Each triangle $\{u, v, w\}$ in a graph contains at least one horizontal-edge.
\label{lem:H-E}
\end{lemma}

\begin{proof}
(Proof by contradiction)
A triangle is a path of length 3 that starts and ends at the same vertex. Suppose there are no horizontal-edges in the triangle. In that case, every edge in the path (i.e., a tree-edge or strut-edge) either increases or decreases the level by one.

Since the path must end on the same level as the starting vertex, the number of edges in the path that decrease the level must be equal to the number of edges that increase the level. Consequently, the length of the path must be even to maintain level parity. However, this contradicts the fact that a triangle has an odd path length of 3.

Therefore, we conclude that there must be at least one horizontal-edge in every triangle.
\end{proof}

\begin{theorem}[Cover-Edge Set Generation]
All horizontal-edges form a valid cover-edge set.
\label{theo:Gen-CES}
\end{theorem}

\begin{proof}
According to Definition \ref{def:coveredge}, for any triangle $\Delta$ in graph $G$, we can always find at least one horizontal-edge that serves as a cover-edge for $\Delta$. Thus, the set of all horizontal-edges constitutes a cover-edge set.
\end{proof}

Therefore, we can construct a cover-edge set, denoted as \emph{BFS-CES}, by selecting all the horizontal-edges obtained during a breadth-first search (\emph{BFS}). It is evident that \emph{BFS-CES} is a subset of $E$ and is typically much smaller than the complete edge set $E$.  

\subsection{Cover-Edge based Triangle Counting}
In this subsection, we provide a comprehensive description of the process involved in identifying all triangles using a cover-edge set generated through a breadth-first search.

\begin{lemma}
Each triangle $\{u, v, w\}$ must contain either one or three horizontal-edges.
\label{lem:HE2}
\end{lemma}

\begin{proof}
By referring to the proof of Lemma~\ref{lem:H-E}, we know that the path corresponding to the triangle's three edges consists of an even number of tree-edges and strut-edges. This implies that there can be either 0 or 2 tree- or strut-edges within each triangle.

In the case where there are 0 tree- or strut-edges, all three edges of the triangle must be horizontal-edges. This is because the absence of tree- or strut-edges implies that the entire path is composed of horizontal-edges.

In the case where there are 2 tree- or strut-edges, the triangle contains exactly one horizontal-edge. This is because having two tree- or strut-edges in the path means that there is one horizontal-edge connecting the remaining two vertices.

Therefore, we conclude that each triangle $\{u, v, w\}$ must contain either one or three horizontal-edges.
\end{proof}

Our triangle counting approach, described in Alg.~\ref{alg:seq}, efficiently counts triangles using a cover-edge set. In line \ref{l:init}, we initialize the counter $T$ to 0, which will store the total number of triangles. To generate the cover-edge set, we perform a breadth-first search (BFS) starting from any unvisited vertex, identifying the 
level ($L(v)$) of each vertex $v$ in its respective component, as shown in lines \ref{l:bfs1} to \ref{l:bfs2}. In lines \ref{l:edge1} to \ref{l:edge2} the algorithm iterates over each edge, selecting the cover-set of horizontal edges $\edge{u,v}$ in a direction-oriented fashion in line~\ref{l:horiz}.
For each vertex $w$ in the intersection of $u$ and $v$'s neighborhoods (line~\ref{l:intersect}), we check the following two conditions to determine if $(u, v, w)$ is a unique triangle to be counted (line~\ref{l:logic}).
If $L(u) \neq L(w)$  then the edge $\edge{u, v}$ is the only horizontal-edge in the triangle $(u, v, w)$. If $L(u) \equiv L(w)$, then the edge $\edge{u,v}$ is one of three horizontal-edges in the triangle $(u, v, w)$. To ensure uniqueness, the algorithm then checks the added constraint that $v<w$. If the constraints are satisfied, we increment the triangle counter $T$ in line~\ref{l:edge2}.

This approach effectively counts the triangles in the graph while avoiding redundant counting.

\begin{algorithm}[htbp]
\footnotesize
\caption{CETC:Cover-Edge Triangle Counting }
\label{alg:seq}
\begin{algorithmic}[1]
\Require{Graph $G = (V, E)$}
\Ensure{Triangle Count $T$}
\State $T \leftarrow 0$  \label{l:init}
\State $\forall v \in V$ \label{l:bfs1}
\State \hspace{8pt} if $v$ unvisited, then BFS($G$, $v$) \label{l:bfs2}
\State $\forall \edge{u, v} \in E$ \label{l:edge1}
\State \hspace{8pt} if $(L(u) \equiv L(v)) \land (u < v)$ \label{l:horiz} \Comment{$\edge{u,v}$ is horizontal}
\State \hspace{16pt} $\forall w \in N(u) \cap N(v)$ \label{l:intersect}
\State \hspace{24pt} if $(L(u)\neq L(w)) \lor \left( (L(u) \equiv L(w)) \land (v<w) \right)$ then \label{l:logic}
\State \hspace{32pt}  $T\leftarrow T + 1 $     \label{l:edge2} 
\end{algorithmic}
\end{algorithm}

\begin{theorem}[Correctness]
Alg.~\ref{alg:seq} can accurately count all triangles in a graph $G$.
\label{theo:correctness}
\end{theorem}
\begin{proof}
Lemma~\ref{lem:HE2} establishes that a triangle in the graph falls into one of two cases: 1) the two endpoint vertices of the horizontal-edge are on the same level while the apex vertex is on a different level, or 2) all three vertices of the triangle are at the same level.

Consider a triangle $\{v_a, v_b, v_c\}$ in $G$. Without loss of generality, assume that $\edge{v_a, v_b}$ is a horizontal-edge, implying $L(v_a) \equiv L(v_b)$. Let $v_c$ be the apex vertex. The two cases can be distinguished as follows:

For the first case, each triangle is uniquely defined by a horizontal-edge and an apex vertex from the common neighbors of the horizontal-edge's endpoint vertices. Whenever Alg.~\ref{alg:seq} identifies such a triangle $\{v_a,v_b,v_c\}$, it increments the total triangle count $T$ by 1.

In the second case, where all three vertices are at the same level ($L(v_c) \equiv L(v_a) \equiv L(v_b)$), Alg.~\ref{alg:seq} ensures that $T$ is increased by 1 only when $v_a < v_b < v_c$. This condition ensures that triangle $\{v_a,v_b,v_c\}$ is counted only once, preventing triple-counting and ensuring the correctness of the triangle count.

Hence, Alg.~\ref{alg:seq} is proven to accurately count all triangles in the graph $G$.
\end{proof}

The time complexity of Alg.~\ref{alg:seq} can be analyzed as follows. The computation of breadth-first search, including determining the level of each vertex and marking horizontal-edges, requires $\mathcal{O}(n+m)$ time.

Since there are at most $\mathcal{O}(m)$ horizontal-edges, finding the common neighbors of each horizontal-edge individually can be done in $\mathcal{O}(d_{\text{max}})$ time. Here, $d_{\text{max}}$ represents the maximal degree of a vertex in the graph.

Therefore, the overall time complexity of Alg.~\ref{alg:seq} is $ \mathcal{O}(m \cdot d_{\text{max}})$.

\section{Communication Efficient Triangle Counting Algorithm}
\label{sec:parallel}
This section presents our communication-efficient parallel algorithm for counting triangles in massive graphs on a $p$-processor distributed-memory parallel computer. We will take advantage of the concept of \emph{Cover-Edge Set} to significantly improve the communication performance  of our triangle counting method.
Since distributed triangle counting is communication-bound \cite{pearce2017triangle}, this algorithm is expected to improve the overall running time.
The input graph $G$ is stored in a compressed sparse row (CSR) format. The vertices are partitioned non-uniformly to the $p$ processors such that each processor stores approximately $2m/p$ edge endpoints.
This graph input follows the format used by the majority of parallel graph algorithm implementations and benchmarks such as Graph500 and Graph Challenge. 

\subsection{Parallel Algorithm Description}
Our communication-efficient parallel algorithm (see Alg.~\ref{algparallel}) is based on the same cover-edge approach proposed in section \ref{sec:new_approach}. %
The binary operator $\oplus$ used in line~\ref{l:swap} is bitwise exclusive OR (XOR).

\begin{algorithm}[htbp]
\footnotesize
\caption{Comm-CETC: Communication Efficient Triangle Counting}
\label{algparallel}
\begin{algorithmic}[1]
\Require{Graph $G = (V, E)$}
\Ensure{Triangle Count $T$}
\State Run parallel BFS($G$) and build partial cover-edge set $S_i$ on $p_i$ \label{l:cbfs}
\State For all $p_i, i \in \{0 \ldots p-1\}$ in parallel do: \label{l:bcast1}
    \State \hspace{8pt} $t_i \leftarrow 0$  \label{l:pinit}
    \State \hspace{8pt}  $\forall \edge{u,v} \in S_i$ with $u<v$ on $p_i$ \label{l:sloop1}
        \State \hspace{16pt} $\forall w \in V_i$ such that $w \in N(u), N(v)$ \label{l:select1}
            \State \hspace{24 pt}  if $(L(u)\neq L(w)) \lor \left( (L(u) \equiv L(w)) \land (v<w) \right)$  then \label{l:count1}
                \State \hspace{32 pt} $t_i = t_i + 1$ \label{l:eloop1}
    \State \hspace{8pt} For $j \leftarrow 1$ to $p-1$ do: \label{l:bgraph}
        \State \hspace{16pt} Processors $i$ and $i \oplus j$ swap edge sets $S_i$ and $S_j$. \label{l:swap}
        \State \hspace{16pt}  $\forall \edge{u, v} \in S_j$ with $u<v$ on $p_i$ \label{l:sloop2}
            \State \hspace{24pt} $\forall w \in V_i$ such that $w \in N(u), N(v)$ \label{l:select2}
                 \State \hspace{32 pt}  if $(L(u)\neq L(w)) \lor \left( (L(u) \equiv L(w)) \land (v<w) \right)$ then \label{l:count2}
                    \State \hspace{40 pt} $t_i = t_i + 1$  \label{l:eloop2}
\State $T \leftarrow$ Reduce$(t_{i}, +)$ \label{l:reduce}
\end{algorithmic}
\end{algorithm}

Similar to the baseline \emph{CETC} algorithm, the cover-edge set $S = \cup_{i=0}^{p-1} S_i$ is determined in line~\ref{l:cbfs} by labeling the horizontal edges from a parallel BFS.

Each processor runs lines~\ref{l:bcast1} to \ref{l:eloop2} in parallel that consists of two main substeps. Local triangles are counted in lines~\ref{l:sloop1} to \ref{l:eloop1} and a total exchange of cover-edges between each pair of processors to count triangles is performed in lines~\ref{l:bgraph} to \ref{l:eloop2}. 
Note at the end of each iteration of the \emph{for} loop, processor $p_i$ can discard the cover-edge set $S_j$.
In lines~\ref{l:select1} and \ref{l:select2}, processor $p_i$ determines for each cover edge $\edge{u, v}$ all the apex vertices $w$ held locally that are adjacent to both $u$ and $v$. The logic for counting triangles in lines~\ref{l:count1} and \ref{l:count2} is similar to Alg.~\ref{alg:seq} as to only count unique triangles. Finally, a reduction operation in line~\ref{l:reduce} calculates the total number of triangles by accumulating the triangle counters across the system, i.e., $T = \sum_{i=0}^{p-1} t_i$.

\subsection{Cost Analysis}

\subsubsection{Space}

In addition to the input graph data structure, an additional bit is needed per edge (for marking a horizontal-edge) and \bigO{\lceil \log D \rceil} bits per vertex to store its level, where $D$ is the diameter of the graph. This is a total of at most $m + n\lceil \log D \rceil$ bits across the $p$ processors. Preserving the graph requires additional \bigO{n+m} space for the graph. 

\subsubsection{Compute}

The BFS costs \bigO{(n+m)/p} \cite{cormen2022algorithms}, the modified neighbor sets take \bigO{m/p}. 
The search corresponding to one cover-edge in a vertex's adjacency list takes at most \bigO{\log(d_{max})} time using binary search, and only \bigO{1} expected time using a hash table. Let $d_i$ be the degree of vertex $v_i$ where $0\le i<n$. Searching $km$ edges in all vertices' adjacency lists takes $\mathcal{O}(km\sum_{i=0}^{n-1} \log(d_i)) = \mathcal{O}(km \log(\Pi_{i=0}^{n-1}d_i))$ time. Since $\sum_{i=0}^{n-1}d_i=2m$, we know that $\log(\Pi_{i=0}^{n-1}d_i)$ reaches its maximum value when $d_i=2m/n$ for $0\le i<n$. Thus, $\mathcal{O}(km \log(\Pi_{i=0}^{n-1}d_i)) \le \mathcal{O}(km \log((2m/n)^n)) \le \mathcal{O}(kmn \log(2n^2/n)) = \mathcal{O}(mn \log(n))$.
\subsubsection{Total Communication}
In our analysis of communication cost for BFS, we measure the total communication volume independent of the number of processors. Thus, this is a conservative overestimate of communication since a fraction (e.g., $1/p$) of accesses will be on the same compute node versus message traffic between nodes. At the same time, we do not consider the savings from overlapping with the computation cost.

The cost of the breadth-first search is $m$ edge traversals with $\lceil \log D \rceil + 3 \lceil \log n \rceil$ bits communicated per edge traversal for the level information, pair of vertex ids, and vertex degree, yielding $m \cdot (\lceil \log D \rceil + 3 \lceil \log n \rceil)$ bits for the BFS.
Transferring $km$ horizontal-edges requires $kmp \lceil \log n \rceil$ bits, where $p$ is the number of processors.
The final reduction to find the total number of triangles requires $(p-1)\lceil \log n \rceil$ bits.

Hence, the total communication volume is 
$m \cdot (\lceil \log D \rceil + 3 \lceil \log n \rceil) + 
kmp \lceil \log n \rceil +
(p-1) \lceil \log n \rceil
=
m \cdot (\lceil \log D \rceil + (kp + 3) \lceil\log n \rceil) + (p -1) \lceil \log n \rceil$
bits.  Hence, since the word size is $\bigTheta{\log n}$ and $D \leq n$, 
the communication is \bigO{pm} words.

\section{Communication Analysis on Real and Synthetic Graphs}

\begin{table*}[!ht]
    \centering
    \caption{Communication costs for real and synthetic graph. The synthetic graphs are Graph500 RMAT graphs of scale 36 and 42. The column \textbf{`Previous'} represents the communication volume of the best prior parallel algorithms \cite{dolev2012tri,pearce2018k,sanders2023engineering}, that use wedge-checking based algorithms and \textbf{`This paper'} represents the communication cost of our new approach. \textbf{`Reduction'} represents the communication reduction between these two, and thus, the expected speedup of the parallel algorithm. Entries in \emph{italics} are estimated values.}
    \label{tab:results}
    \begin{tabular}{|l|r|r|r|r|r|r|r|r|r|}
    \hline
        \textbf{Graph} & \textbf{n} & \textbf{m} & \textbf{\# Triangles} & \textbf{\# Wedges} & \textbf{$k$} & \textbf{$p$} & \textbf{Previous} & \textbf{This paper} & \textbf{Reduction} \\ \hline
        ca-GrQc & 5242 & 14484 & 48260 & 165798 & 0.522 & 4 & 526KB & 122KB & 4.31 \\ \hline
        ca-HepTh & 9877 & 25973 & 28339 & 277389 & 0.423 & 4 & 948KB & 218KB & 4.35 \\ \hline
        as-caida20071105 & 26475 & 53381 & 36365 & 776895 & 0.225 & 4 & 2.78MB & 401KB & 7.10 \\ \hline
        facebook\_combined & 4039 & 88234 & 1612010 & 17051688 & 0.914 & 4 & 48.8MB & 893KB & 56.0 \\ \hline
        ca-CondMat & 23133 & 93439 & 173361 & 1567373 & 0.511 & 4 & 5.61MB & 897KB & 6.40 \\ \hline
        ca-HepPh & 12008 & 118489 & 3358499 & 5081984 & 0.621 & 4 & 17.0MB & 1.13MB & 15.1 \\ \hline
        email-Enron & 36692 & 183831 & 727044 & 5933045 & 0.478 & 4 & 22.6MB & 1.79MB & 12.7 \\ \hline
        ca-AstroPh & 18772 & 198050 & 1351441 & 8451765 & 0.667 & 4 & 30.2MB & 2.08MB & 14.6 \\ \hline
        loc-brightkite\_edges & 58228 & 214078 & 494728 & 6956250 & 0.441 & 4 & 26.5MB & 2.02MB & 20.4 \\ \hline
        soc-Epinions1 & 75879 & 405740 & 1624481 & 21377935 & 0.498 & 4 & 86.7MB & 4.25MB & 10.7 \\ \hline
        amazon0601 & 403394 & 2443408 & 3986507 & 96348699 & 0.529 & 8 & 436MB & 40.9MB & 10.7 \\ \hline
        
        com-Youtube & 1134890 & 2987624 & 3056386 & 209811585 & 0.347 & 8 & 1.03GB & 44.3MB & 23.7 \\ \hline
        RMAT-36 & 68719476736 & 1099511627776 & \emph{1.2E+14} & \emph{2.73E+16} & \emph{0.311} & 128 & 218PB & \emph{192TB} & \emph{1156} \\ \hline
        RMAT-42 & 4398046511104 & 70368744177664 & \emph{1.3E+16} & \emph{5.79E+18} & \emph{0.260} & 256 & 52.8EB & \emph{22.8PB} & \emph{2368} \\ \hline
    \end{tabular}
\end{table*}
In this section, we analyze the performance of the parallel triangle counting algorithm on both real and synthetic graphs. We implemented our new triangle counting algorithm using Python to accurately compute the exact communication volume and determine an analytic model based on the size of the graph and number of processors, and the ratio or percentage ($k$) of cover-edges from the BFS.  The results given in Table~\ref{tab:results} are exact communication volumes from our new algorithm on all of the graphs except the two large RMAT graphs where we compute the communication volume from the validated analytic model. For the comparison with prior approaches \cite{dolev2012tri,pearce2018k,sanders2023engineering}, we estimate the communication volume from the number of wedges which is exact for all graphs other than the last two large RMAT graphs where we estimate the number of wedges using graph theory.

For the real graphs, we find the actual value of $k$, the percentage of graph edges that are cover-edges, for an arbitrary breadth-first search, and set the number $p$ of processors to a reasonable number given the size of the graph.  For the synthetic graphs, we use large Graph500 RMAT graphs \cite{chakrabarti2004r} with parameters $a = 0.57$, $b = 0.19$, $c = 0.19$, and $d = 0.05$, for scale 36 and 42 with  $n=2^{\mbox{scale}}$ and $m=16n$, similar with the IARPA AGILE benchmark graphs, and set $p$ according to estimates of potential system sizes with sufficient memory to hold these large instances. 

For comparison, most prior parallel algorithms for triangle counting operate on the graph as follows.
A parallel loop over the vertices $v \in V$ produces all 2-paths (\emph{wedges}) where $\edge{v,v_1} , \edge{v,v_2} \in E$ and (w.l.o.g.) $v_1<v_2$. The processor that produces this wedge will send an open wedge query message containing the vertex ids of $v_1$ and $v_2$ to the processor that owns vertex $v_1$. If the consumer processor that receives this query message finds an edge $\edge{v_1,v_2} \in E$, then a local triangle counter is incremented. After producers and consumers complete all work, a global reduction over the $p$ triangle counts computes the total number of triangles in $G$.

\subsection{Graph500 RMAT Graphs}

\begin{figure}
    \centering
    \includegraphics[width=0.5\textwidth]{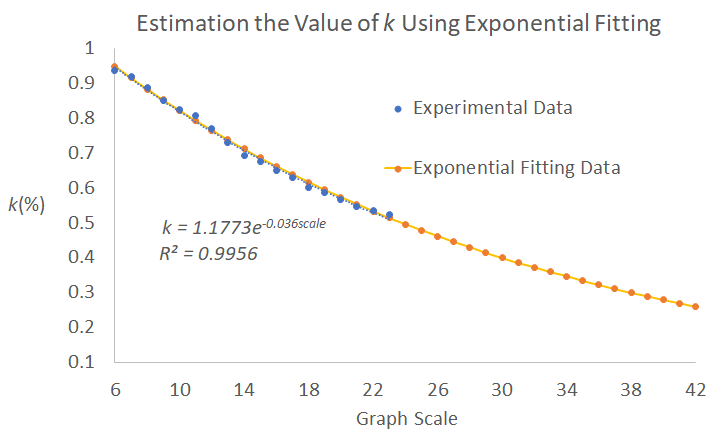}
    \caption{Estimate of $k$ using an exponential model, based on observations of $k$ for RMAT graph scale 6 to 23 graphs.}
    \label{fig:k}
\end{figure}

For the large  Graph500 RMAT graphs, the number of triangles is estimated from our model based on the number of triangles found in RMAT graphs up to scale 29 in the literature \cite{hoang2019disttc,giechaskiel2015pdtl,chakrabarti2004r,burkhardt2017graphing}. The fitting equation is $\mbox{\#Triangles} = 77.422n^{1.125}$ with $R^2 = 1.0$, where $n$ is the total number of vertices. The number of triangles estimated for scale 36 and 42 RMAT graphs are $1.20 \times 10^{14}$ and $1.30 \times 10^{16}$, respectively.

We estimate the number of wedges for the scale 36 and 42 Graph500 RMAT graphs based on the theorem given by Seshadhri \emph{et al.} in \cite{seshadhri2011hitchhiker}. According to their formula, we can estimate the expected number of vertices \(N(d)\) for a given out-degree \(d\). The number of wedges that can be formed by vertices with such a degree is calculated as \(\binom{d}{2} \times N(d)\), where \(\binom{d}{2}\) means choosing two from \(d\).

By summing all such wedges generated from the minimum ($ e\ln n $)  to the maximum degree (\(\sqrt{n}\)), which is the assumption of the formula,  we can approximate the total number of wedges in the given graph, where \(n\) is the total number of vertices. This is a conservative estimate because it only considers the out-degree instead of the sum of out and in-degrees. Employing the formula, we calculate the number of wedges to be \(2.73 \times 10^{16}\) for scale 36 and \(5.8 \times 10^{18}\) for scale 42. With $2\log n$ bits/wedge, the total volume of wedge checks is 218PB and 52.8EB for RMAT graphs of scales 36 and 42, respectively\footnote{Throughout this paper, a petabyte (PB) is $2^{50}$ bytes and an exabyte (EB) is $2^{60}$ bytes.}.


Beamer~\emph{et~al.} \cite{beamer2011searching} find a typical BFS on a scale 27 Graph500 RMAT graph has 7 levels, so 4 bits is a reasonable estimate for $\log D$ in our analyses of scale 36 and 42 graphs. 


\looseness=-1
The methodology for estimating the value of $k$ for RMAT graphs is as follows. RMAT graphs from scale 6 to 23 are generated, and the exact value of $k$ is determined for each by counting the horizontal-edges after a breadth-first search. The data fit to an exponential model $k = 1.1773 e ^{-0.036 \cdot \mbox{scale}}$ with very high $R^2 = 0.9956$ (see Fig.~\ref{fig:k}). For scale 36, $k$ is estimated to be 0.311 and for scale 42, $k$ is estimated to be 0.260.

In our new approach for scale 36, where the communication cost is 
$m \cdot (\lceil \log D \rceil + (kp + 3) \lceil\log n \rceil) + (p-1) \lceil \log n \rceil$ bits. With $\lceil \log D \rceil = 4$, and assuming $p=128$ processors, we have a total communication volume of 192TB, for a communication reduction of $1156 \times$. 
For scale 42, and assuming $p=256$ processors, we estimate the communication of our new triangle counting algorithm as 22.8PB, for a communication reduction of $2368 \times$.

\section{Related Work}
\label{sec:related_work}

\subsection{Sequential Algorithms}

The na\"ive approach for triangle counting uses brute-force: find all the triplets $\{v_a,v_b,v_c\}$, that is, permutations of three arbitrary vertices in the graph, and check whether each edge in the triplet exists. The time complexity is $\bigOmega {n^3}$. Latapy \cite{latapy2007practical} and Schank and Wagner \cite{schank2005finding} provide surveys of faster sequential algorithms. Triangle counting generally can be formulated as three kinds of problems: set (list) intersection, matrix multiplication and subgraph (cycle) query.

The three main \textbf{intersection-based} triangle counting algorithms are: 
1) the node-iterator algorithm iterates over all vertices and tests for each pair of neighbors whether they are connected by an edge, 2) the edge-iterator algorithm iterates over all edges and searches for common neighbors of the two endpoints of each edge, 
and 3) the forward algorithm is a refinement of the edge-iterator algorithm that computes the intersection of a subset of neighborhoods by using an orientation of the graph. 
The time complexity of node-iterator and edge-iterator are both \bigO{m \cdot d_{\mbox{max}}} and the forward algorithm is \bigO{m^{\frac{3}{2}}}, which has significantly better performance when 
$d_{\mbox{max}} \gg \sqrt{m}$
\cite{latapy2007practical}.

When performing the intersection of two lists, the commonly used techniques are \emph{merge-path}, \emph{binary search} and \emph{hashing-based} algorithms. 
 
Merge-path algorithms (e.g., \cite{green2014fast, shun2015multicore}) use two pointers to scan through neighbor lists of two endpoints from beginning to end in order to find the list intersection. During the scan, the pointer that points to a smaller value will be incremented. A triangle is enumerated if both pointers are incremented (i.e., they both point to the same vertex). 
Binary-search algorithms (e.g., \cite{hu2018tricore, hoang2019disttc}) organize the longer list as a binary tree and use the shorter list as search keys. For each search key, it descends through the binary-search tree in order to find the equal entry, which is a triangle. 
Hashing-based algorithms (e.g., \cite{pandey2021trust, shun2015multicore}) construct a hash table for one list and use the other list as search keys to find the common elements in the hash table. The hash table is used here to find the intersection of two adjacency lists, so it is not necessary to sort all the adjacency lists to find all the triangles. The running time is proportional to the size of the two adjacency lists.

Triangle counting using \textbf{matrix multiplication} \cite{azad2015parallel} relies on a linear algebra formulation for triangle counting. 
This approach can be optimized \cite{acer2019scalable} using matrix decomposition by decomposing $A$ into lower and upper triangular matrices $L$ and $U$, and then computing $(L \times U) \odot L$, or $(L \times L) \odot L$ to determine the number of triangles. The binary operator $\odot$ denotes the Hadamard product.

A \textbf{subgraph-based} approach for triangle counting  searches for all occurrences of a query graph, which is a triangle, in the input graph. Wang and Owens \cite{wang2019fast} use breadth-first search to update the subgraph matching approach by pruning more invalid vertices based on neighborhood encoding information, and using optimizations like $k$-step look-ahead to reduce unwanted intermediate results. Alon \emph{et~al.} \cite{alon1997finding} proposed a \bigO{m^{1.41}} algorithm to find length 3 cycles (triangle) in a graph, which is an improvement over the Itai and Rodeh sequential \bigO{m^{\frac{3}{2}}} algorithm \cite{itai1977}.

\subsection{Parallel Algorithms}
 
Map-reduce is a standard platform for large scale distributed computation. Cohen \cite{cohen2009graph} first demonstrated the capability of map-reduce to solve triangle counting in an approach that generates \emph{open wedges} between triples of vertices in the graph and determines if a closing edge exists that completes a triangle. Suri \emph{et~al.} \cite{suri2011counting} implemented triangle counting using map-reduce that ranks vertices by degree and distributes them across hosts. Pearce \cite{pearce2017triangle} developed an algorithm that is based on creating an augmented degree-ordered directed graph, where the original undirected edges are directed from low-degree to high degree, and implemented this approach in the distributed asynchronous graph processing framework HavoqGT. DistTC \cite{hoang2019disttc} is a distributed triangle counting implementation for multiple machines that uses mirror proxy on each partition to eliminate almost all the inner-host communication. TriCore \cite{hu2018tricore} partitions the graph held in a compressed-sparse row (CSR) data structure for multiple GPUs and uses stream buffers to load edge lists from CPU memory to GPU memory on-the-fly and then uses binary search to find the intersection.  Hu \emph{et~al.} \cite{hu2021accelerating} employed a ``copy-synchronize-search” pattern to improve the parallel threads efficiency of GPU and mixed the computing and memory intensive workloads together to improve the resource efficiency. Pandey \emph{et~al.} \cite{pandey2021trust} employed an vertex-centric hash-based design to scale triangle counting to over 1,000 GPUs.  TriC \cite{ghosh2020tric} exploits the vertex-based distributed triangle counting and sends vertices rather than edges (vertex pairs), and then the remote processor could translate the sequence of vertex IDs to correct combination of vertices as edges to reduce communication. An enhancement is then presented to TriC \cite{ghosh2022improved} that added a user-defined buffer to improve the flexibility of controlling the memory usage for large data sets and used a probabilistic data structure to optimize the edge lookups by trading off the accuracy. Strausz \emph{et~al.} \cite{Strausz22} use CLaMPI, a software caching layer that caches data retrieved through MPI remote memory access operations, to reduce the overall communication cost. Zeng \emph{et~al.} \cite{zeng2022htc} proposed a triangle counting algorithm that adaptively selects vertex-parallel and edge-parallel paradigm.

Panduranga \emph{et~al.} \cite{pandurangan2021distributed} and Dolev \emph{et~al.} \cite{dolev2012tri}'s work focused on the communication cost. Compared with our work, there are two major differences. First, they use the number of communication rounds to measure the total communication with a bandwidth restriction. However, we use the total volume of messages to evaluate the communication. 
Second, they are probabilistic algorithms, but our algorithm is a deterministic algorithm (Dolev~\emph{et al.}~\cite{dolev2012tri} also contains a deterministic version). Probabilistic methods cannot be used under scenarios with an exact result requirement.  Uhl \cite{uhl2021communication,sanders2023engineering} also focuses on reducing the communication cost of triangle counting. The paper's basic idea is only requiring communication for counting triangles consisting of cut edges. If the partition generates many cut edges, the proposed method cannot significantly reduce communication. In contrast, our method identifies a subset of the total set of edges independent of the partitioning and only transfers this smaller set of edges during the triangle counting to significantly reduce the total communication.

\section{Conclusions}
\label{sec:conclusion}

In this paper, we present novel sequential and parallel algorithms for counting and finding triangles in graphs based on a compact cover-edge set. The parallel algorithm is the first communication-efficient triangle counting algorithm by exploiting BFS horizontal-edges to significantly reduce the communication volume on massive graphs of practical interest.
Our approach uses the breadth-first search to significantly reduce the number of edges examined and minimize the communication required for triangle checking. The parallel algorithm achieves an order of magnitude or more reduction of communication volume for large graphs as communication is the main bottleneck for triangle counting on distributed memory systems.

\section{Reproducibility}
\label{sec:reproducibility}

The sequential triangle counting source code and the Python code for determining the communication volume of the parallel algorithm are open source and available on GitHub at \url{https://github.com/Bader-Research/triangle-counting}.  The input graphs are from the Stanford Network Analysis Project (SNAP) available from \url{http://snap.stanford.edu/}.

\bibliographystyle{IEEEtran}
\bibliography{ref}

\end{document}